\documentclass[11pt]{amsart}

\newcommand{\FBF}{\textsc{FBF}}

\usepackage{amsmath,amssymb,amsthm}
\usepackage{mathtools}

\usepackage{hyperref}
\usepackage[nameinlink]{cleveref}

\usepackage{paralist}

\usepackage{booktabs}
\usepackage{threeparttable} %% checked is Overleaf compliant

\usepackage{tabularx}
\usepackage{array}
\newcolumntype{L}{>{\raggedright\arraybackslash}p{0.15\textwidth}}
\newcolumntype{C}{>{\centering\arraybackslash}p{0.40\textwidth}}

\usepackage{graphicx}
\usepackage{tikz}

\title[Method for Ramsey Colourings for 8 and 13 Colours]{%%Fusion Beautiful Fusion: \\ 
A New Method that can Generate Ramsey Colourings for Eight and Thirteen Colours}

\author[C. Gretton, T. Kowalski, R. Pennifold, and C. Christopher]{Charles Gretton} 
\thanks{Corresponding author Email: charles.gretton@anu.edu.au}

\author[]{Tomasz Kowalski} \author[]{Richard Pennifold} \author[]{Cody Christopher}
\date{\today}

\begin{document}

\begin{abstract}
We study representations for relation algebras corresponding to certain edge colourings of complete graphs. 
Previously suitable colourings were obtained for the number of colours $n$ up to $2000$, with two exceptions: $n = 8$ and $n = 13$. 
Using a method that we have called the \FBF\  (Fusion Beautiful Fusions) method, we find colourings for $8$ and $13$ colours.
Our method is a new guess-and-check approach that we describe as an adaptation of Comer's finite-field method. 
Using \FBF\  we construct novel colourings that are non-isomorphic to existing published colourings for 5, 6, 7, 9, 10, 11 and 12 colours.  
\end{abstract}

\maketitle

\section{Introduction}

%% from Tomasz' paper

For a finite set of $n$ colours $C = \{c_1,\ldots, c_n\}$, and a complete graph
$K_m$=(V,E), can the edges $E$ of $K_m$ be coloured by colours from $C$ in such a way that:

\begin{enumerate}
\item there are no monochromatic triangles, and 
\item every non-monochromatic triangle appears everywhere it can?
\end{enumerate}

\noindent The first condition gives an upper bound for the possible size of $K_m$, via
Ramsey's theorem. The second condition decomposes into the following: \begin{inparaenum}[(i)]\item every vertex of $K_m$ is incident with at least one edge of each colour, and \item given any edge $\{x, y\}$ of colour $c_i$, and any colours $c_j$, $c_k$, such
that $j$ and $k$ are not both equal to $i$, there exists a vertex $z$ such that $\{x, z\}$ is
coloured by $c_j$ and $\{z, y\}$ is coloured by $c_k$.\end{inparaenum} 

Formally, let $c:E\to C$ be a surjective function that assigns a colour in $C$ to each edge of $E$. 
A Ramsey colouring satisfies:

\begin{equation*} 
\begin{aligned}  & \forall c_i \in \operatorname{im}(c),\; \forall \{x,y\}\in E \;\text{ with }\; c(\{x,y\})=c_i,\\ &\forall c_j,c_k \in \operatorname{im}(c) \;\text{ such that }\; |\{c_i,c_j,c_k\}|>1,\\ &\exists z \in V\setminus\{x,y\} \;\text{ such that }\; c(\{x,z\})=c_j \text{ and } c(\{z,y\})=c_k. 
\end{aligned} 
\end{equation*}

Our contribution is to find the first colourings for \mbox{$|C|=8$} and \mbox{$|C|=13$}. We do this with a general-purpose approach which we describe as an adaptation of Comer's finite-field method. 

\section{Related Work}

Existing computational investigations of finite-field constructions using Comer's method have been published by Alm and Manske~\cite{Alm:and:Manske:2015,Alm2017} and Kowalski~\cite{Kowalski2015}.
We depart from the method studied by those authors by {\em fusing} pairs of fine relations, so-called {\em cyclotomic classes}, into coarser relations.
The term ``fusion'' here refers to combining relations as described by Bannai and Ito~\cite{BannaiIto1984}.
We also evaluate non-isomorphism between colourings, with our results established using an incidence-graph reduction with colour vertices and {\sc nauty/Traces}, the tool by McKay and Piperno~\cite{MCKAY201494}.

\section{Background}

\newcommand{\F}{\mathbb F}
\newcommand{\Z}{\mathbb Z}

For a prime power $q=p^e$, $\F_q$ denotes the field with $q$ elements and $\F_q^\times=\F_q\setminus\{0\}$ its cyclic multiplicative group. 
A generator $g$ is a primitive element, which in the case of a prime field is a primitive root modulo $q$.
Assuming $q \equiv 1 \pmod r$, we have $r$ cosets available taking $H=\langle g^{r}\rangle$, specifically:
\[C_a=g^aH,\qquad a\in\Z_r=\Z/r\Z.\]
%%
%%n=9
\noindent These ideas can be made concrete with a small example.
Take $p=449$ and $e=1$ so that $q=p$. We have primitive root $g=3$ and using this can construct $r=16$ cosets with $H=\langle g^{16}\rangle$ as
\[C_a=g^aH=\{g^{a+16t}:t\in\Z, 0\le t<(q-1)/16\}.\]
\noindent We refer to $C_a$ as a {\em fine} set, and by fusing such fine cosets into {\em coarser} sets shall arrive at partitions suitable for our desired colourings. 

\section{\FBF\  Method:\\Adapting Comer's Finite-Field Method for 8 and 13 Colours}

We suppose that $p$, $g$ and the $r$ cosets $\{C_0,\ldots, C_{r-1}\}$, as above, are available.
%%
%% Although the method we describe can be applied to prime powers $q=p^e$, we restrict our search to prime fields. 
%%
An appropriate prime $p$ is discovered by iterating over primes in order, starting with a sufficiently small prime.
Here, we shall always take $q=p$ and $e=1$, and leave explorations of prime powers to future work.
We choose $g$ to be a generator of $\F^\times_p$.
The value of $r$ is chosen to be some positive integer multiple of $n$, the desired number of colours. 
We only consider $r=2n$ below because this gives us the desired colourings and it keeps the exposition simple. 
More generally, with $r=sn$, the $r$ fine cosets can be partitioned into $n$ blocks of $s$ cosets. 
Fusing the cosets in each block produces $n$ coarse sets, each of cardinality $(q-1)/n$.
%%
%%More generally, for $r=sn$, the method can be adapted to merge the fine cosets into coarse sets of size $s$ to arrive at an appropriate partition. 
For example, $s=3$ would require merging triples, $s=4$ would require quadruples, 
and so on.
%%
%% Our investigations thus far with fusions over triples, quadruples, etc., have not been fruitful. 

With the above in hand, at prime $p$ an iteration of our search method proceeds serially as follows.\footnote{A {\sc SageMath} notebook demonstrating the approach is available here \url{https://github.com/charlesgretton/ramsey-colourings}.}

\begin{enumerate}
\item For every possible pair-fusion $D_{a,b}=C_a\cup C_b$ of fine cosets with $a,b\in\Z_r$ and $a\neq b$, keep only sets $D_{a,b}$ satisfying: 
    \begin{enumerate}[(i)]
      \item The inverse closure rule $-D_{a,b}=D_{a,b}$, and 
      \item the condition ensuring no monochromatic triangles and the existence of every permitted triangle containing two edges of the same colour, specifically $D_{a,b}+D_{a,b}=\mathbb F_p\setminus D_{a,b}$. Here, $A,B\subseteq\mathbb F_p$, the sumset notation above is defined as 
      \[A+B=\{a+b:a\in A, b\in B\}.\]
    \end{enumerate}
\item Construct a compatibility graph $G^D=(V^D,E^D)$ with elements in $V^D$ in one-to-one correspondence with surviving fused sets, and an edge $\{D_{a,b},D_{c,d}\}$ iff $\{a,b\}\cap\{c,d\}=\emptyset$ and ensuring every non-monochromatic triangle will appear everywhere it can via $D_{a,b}+D_{c,d}=\mathbb F_p\setminus\{0\}$. 
\item If a clique of size $n$ in $G^D$ exists, let $\{D_{a_1,b_1},\ldots,D_{a_n,b_n}\}$ be the fused sets corresponding to that clique's vertices. These sets are pairwise disjoint and collectively contain all $r=2n$ fine cosets. They therefore partition $\F^\times_p$. The Ramsey graph colouring of $K_p$ with vertices in one-to-one correspondence with $\F_p$ is then achieved by colouring every edge $\{x,y\}$, for distinct $x,y\in\F_p$, using the rule \[ c(\{x,y\})=c_{a_i,b_i} \iff y-x\in D_{a_i,b_i}.\] If no such clique exists, the check fails for $p$ and search iterates continuing at the next prime.
\end{enumerate}

\subsection{Reduction to index differences and mixed pairings} \label{subsec:usesymmetry}

We describe and exploit a symmetry that reduces the search required to obtain the desired colourings by fusing pairs.
We find that in many cases the desired colouring can be constructed directly, without building and searching the full compatibility graph.
%%
%% This is achieved by first using a filtering step that implements a symmetry reduction.
%%
The index difference $\delta=b-a\pmod{2n}$ of a fused pair $D_{a,b}=C_a\cup C_b$ determines whether it satisfies the required inverse-closure and monochromatic sumset conditions.
Alone, $\delta$ does not determine whether the sumset compatibility conditions between distinct fused pairs are satisfied.
Note that multiplication by $g^{-a}$ maps $D_{a,b}$ to $D_{0,b-a}$.
Specifically, \[ g^{-a}D_{a,b} = g^{-a}(C_a\cup C_b) = C_0\cup C_{b-a} = D_{0,b-a}, \] where the coset indices are taken modulo $2n$.
Multiplication by a nonzero field element distributes over addition, \[ g^{-a}(X+Y)=g^{-a}X+g^{-a}Y \] for all $X,Y\subseteq\F_p$. 
Therefore, multiplication preserves the equalities and containments appearing in the inverse-closure and monochromatic sumset conditions for an individual fused pair.
Simultaneously translating both coset indices by the same element of $\Z_{2n}$ preserves the monochromatic admissibility of a fused pair.
Further, these fused pairs are unordered: \[ D_{a,b}=D_{b,a}. \]
Thus, differences $\delta$ and $-\delta$ describe the same orbit of unordered fused pairs.
The nonzero differences modulo $2n$ can therefore be represented by \[ \delta\in\{1,\ldots,n\}. \]
Only $n$ unoriented difference classes need to be tested, rather than all \[ \binom{2n}{2} \] pairs of fine cosets.
The difference $\delta=n$ is self-inverse under negation modulo $2n$, since \[ -n\equiv n\pmod{2n}, \] whereas the remaining nonzero differences occur in distinct pairs $\{\delta,-\delta\}$.

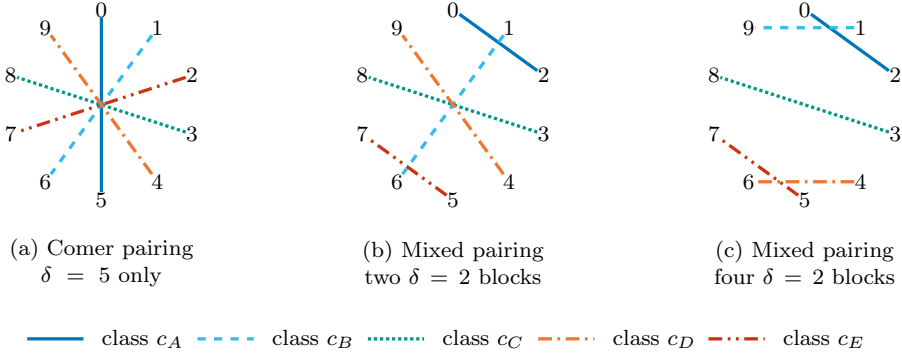
\begin{figure}[t]
  \centering

  % Colour-blind-friendly palette.
  \definecolor{cA}{RGB}{0,119,187}
  \definecolor{cB}{RGB}{51,187,238}
  \definecolor{cC}{RGB}{0,153,136}
  \definecolor{cD}{RGB}{238,119,51}
  \definecolor{cE}{RGB}{204,51,17}

  \resizebox{\linewidth}{!}{%
  \begin{tikzpicture}[
    vertex/.style={
      inner sep=0pt,
      font=\scriptsize
    },
    classA/.style={
      draw=cA,
      very thick,
      solid
    },
    classB/.style={
      draw=cB,
      very thick,
      dashed
    },
    classC/.style={
      draw=cC,
      very thick,
      densely dotted
    },
    classD/.style={
      draw=cD,
      very thick,
      dash pattern=on 5pt off 2pt on 1pt off 2pt
    },
    classE/.style={
      draw=cE,
      very thick,
      dash pattern=on 5pt off 2pt on 1pt off 2pt on 1pt off 2pt
    }
  ]

    % Panel (a): Comer pairing.
    % The order cA,cB,cE,cC,cD keeps the corresponding fused
    % colour classes consistent with panels (b) and (c).
    \begin{scope}[xshift=-4.5cm]
      \foreach \i in {0,...,9} {
        \pgfmathsetmacro{\ang}{90-36*\i}
        \node[vertex] (a\i) at (\ang:1.22cm) {\i};
      }

      \draw[classA] (a0) -- (a5); % cA
      \draw[classB] (a1) -- (a6); % cB
      \draw[classE] (a2) -- (a7); % cE
      \draw[classC] (a3) -- (a8); % cC
      \draw[classD] (a4) -- (a9); % cD

      \node[
        font=\scriptsize,
        align=center,
        text width=3.0cm
      ] at (0,-2.02cm) {
        (a) Comer pairing\\
        $\delta=5$ only
      };
    \end{scope}

    % Panel (b): mixed pairing with two difference-2 blocks.
    \begin{scope}
      \foreach \i in {0,...,9} {
        \pgfmathsetmacro{\ang}{90-36*\i}
        \node[vertex] (b\i) at (\ang:1.22cm) {\i};
      }

      \draw[classA] (b0) -- (b2); % cA
      \draw[classB] (b1) -- (b6); % cB
      \draw[classC] (b3) -- (b8); % cC
      \draw[classD] (b4) -- (b9); % cD
      \draw[classE] (b5) -- (b7); % cE

      \node[
        font=\scriptsize,
        align=center,
        text width=3.0cm
      ] at (0,-2.02cm) {
        (b) Mixed pairing\\
        two $\delta=2$ blocks
      };
    \end{scope}

    % Panel (c): mixed pairing with four difference-2 blocks.
    \begin{scope}[xshift=4.5cm]
      \foreach \i in {0,...,9} {
        \pgfmathsetmacro{\ang}{90-36*\i}
        \node[vertex] (c\i) at (\ang:1.22cm) {\i};
      }

      \draw[classA] (c0) -- (c2); % cA
      \draw[classB] (c1) -- (c9); % cB
      \draw[classC] (c3) -- (c8); % cC
      \draw[classD] (c4) -- (c6); % cD
      \draw[classE] (c5) -- (c7); % cE

      \node[
        font=\scriptsize,
        align=center,
        text width=3.0cm
      ] at (0,-2.02cm) {
        (c) Mixed pairing\\
        four $\delta=2$ blocks
      };
    \end{scope}

    % Shared key. Colour and line pattern both identify the class.
    \begin{scope}[xshift=-5.45cm,yshift=-3.00cm]
      \draw[classA] (0,0) -- (0.72,0);
      \node[anchor=west,font=\scriptsize]
        at (0.82,0) {class $c_A$};

      \draw[classB] (2.18,0) -- (2.90,0);
      \node[anchor=west,font=\scriptsize]
        at (3.00,0) {class $c_B$};

      \draw[classC] (4.36,0) -- (5.08,0);
      \node[anchor=west,font=\scriptsize]
        at (5.18,0) {class $c_C$};

      \draw[classD] (6.54,0) -- (7.26,0);
      \node[anchor=west,font=\scriptsize]
        at (7.36,0) {class $c_D$};

      \draw[classE] (8.72,0) -- (9.44,0);
      \node[anchor=west,font=\scriptsize]
        at (9.54,0) {class $c_E$};
    \end{scope}

  \end{tikzpicture}%
  }
  \caption{
    %% Colour and line pattern redundantly identify the five colour classes.   
    %%
    Why the fusion search produces colourings outside the standard Comer constructions. 
    For five colours over $\F_{101}$, the standard construction uses the uniform difference-$5$ pairing in~(a).
    Panel~(b) contains two difference-$2$ blocks and three difference-$5$ blocks, while panel~(c) contains four difference-$2$ blocks and one difference-$5$ block.
    These mixed pairings break the uniform cyclic structure of the standard construction. 
    %%
    %% AddedFor example, at $n=9$ and $p=397$, exhaustive enumeration gives $3{,}028$ translation-orbit representatives, whose coloured-triangle invariants are all distinct. Those colourings are pairwise non-isomorphic including if we allow permutations of colours.
  }\label{fig:nonisomorphic-mixed-pairings}
\end{figure}

Let $H_{\mathrm{adm}}$ be the graph on vertex set $\Z_{2n}$ in which $\{a,b\}$ is an edge precisely when the difference $b-a\pmod{2n}$ is admissible, i.e., when $D_{a,b}$ is a vertex of $G^D$.
Thus, each edge $\{a,b\}$ represents an individually admissible fused pair $D_{a,b}=C_a\cup C_b$.
A perfect matching of $H_{\mathrm{adm}}$ is a set of $n$ pairwise disjoint edges covering every vertex of $\Z_{2n}$ exactly once. 
Equivalently, it specifies $n$ disjoint fused pairs that cover all $2n$ fine cosets.
%%
%%A collection of $n$ disjoint fused pairs covering all $2n$ fine cosets is precisely a {\em perfect matching} of $H_{\mathrm{adm}}$.
%%
A relatively efficient search may therefore proceed in two stages:
\begin{enumerate}
\item enumerate perfect matchings of $H_{\mathrm{adm}}$; and
\item for each perfect matching, verify the sumset compatibility condition \[ D_{a,b}+D_{c,d}=\F_p\setminus\{0\} \] for every two distinct selected fused pairs.
\end{enumerate}
This separates the inexpensive, symmetry-reduced generation of monochromatically admissible pairs from the more restrictive compatibility tests between distinct colour classes.
The existence of such a perfect matching is only a necessary condition and does not, in general, establish that the corresponding fusion is a Ramsey colouring.

%% Equivalently, a proposed collection of $n$ disjoint fused pairs covering all $2n$ fine cosets determines a perfect matching of $H_{\mathrm{adm}}$. 
%%
In an implementation of our construction, these pairs may still be represented as vertices of the compatibility graph $G^D$, with the sumset compatibility conditions between distinct fused pairs encoded by adjacency in $G^D$. 
The difference reduction decreases the work required to construct the candidate vertices before the clique search is performed.

\section{Constructed Ramsey Colourings}
%% give details of the constructed graph and URL links to pajek files. 

We report existing known colourings alongside new colourings found using the construction developed in this work. 
Our search and results only considered primes smaller than $2\cdot 10^3$.
Results are reported in Table~\ref{tab:results}.

\begin{table}[!ht]
\centering
\begin{threeparttable}
\caption{\label{tab:results}
Reporting constructions found using Comer's finite-field method and our fused-pairs method for numbers of colours ranging from $2$ to $13$. 
First column is the number of colours $n$, second column reports values of $p$ at which colourings are found using Comer's method, and last column reports values of $p$ that yield colourings using \FBF. 
The middle column results were known prior to this work.%% obtained without pair-fusions, where each colour is from one of $n$ cosets. 
%%
%%The final column is with pair-fusions, where $2n$ cosets are paired into $n$ sets that represent $n$ colours. 
%%
}
\begin{tabular}{L C C}
\toprule $n$ & no-fusion (historical) & pair-fusion \\
\midrule
2 & $5$ & -- \\
3 & $13$ & $13^=$ \\
4 & $41$ & -- \\
5 & $71,\; 101$ & $\mathbf{101}^{\dagger}$ \\
6 & $97,\; 157,\; 277$ & $\mathbf{97}^{\dagger},\; \mathbf{193},\; \mathbf{313}$ \\
7 & $491$ & $\mathbf{281}$ \\
8 & -- & $\mathbf{449},\ \mathbf{929}$ \\
9 & $523,\ 577$ & $\mathbf{397},\ \mathbf{577}^{\dagger}$ \\
10 & $1181$ & $\mathbf{641}$ \\
11 & $947,\ 1409$ & $\mathbf{1277},\ 1409^=$ \\
12 & $769,\ 1201$ & $\mathbf{769}^{\dagger},\ \mathbf{1201}^{\dagger},\ \mathbf{1249}$ \\
13 & -- & $\mathbf{1613}$ \\
\bottomrule
\end{tabular}
\begin{tablenotes}[flushleft]
\item We use a bold font to indicate a new colouring, that is either for a prime that does not yield a colouring without fusing pairs or that produces a colouring that is not isomorphic to a historical colouring at that prime.
\item We indicate that a pair-fusion colouring was found at a prime at which a historical colouring exists with $^\dagger$. 
\item We use $^=$ if the fused-pairs method finds a historical colouring.
\item A dash symbol `--' indicates the method does not produce a colouring. 
\hrule
\end{tablenotes}
\end{threeparttable}
\end{table}

\newtheorem{theorem}{Theorem}
\begin{theorem} 
Ramsey algebras with \(8\) and \(13\) colours have finite representations. 
In particular, the \FBF\  construction gives representations over \(\mathbb F_{449}\) and \(\mathbb F_{1613}\), respectively. 
\end{theorem}
\begin{proof} 
Table~\ref{tab:results} reports \FBF\  constructions over \(\F_{449}\) and \(\F_{1613}\) for $8$ and $13$ colours, respectively. 
In each case, the corresponding fused sets partition $\F_p^\times$ and satisfy the inverse-closure and sumset conditions required.
The resulting edge colourings provided correspond to the claimed representations. 
\end{proof}

\newcommand{\Can}{\operatorname{Can}}

\section{Coloured Isomorphism Testing}

In Table~\ref{tab:results} we report a number of new colourings that are not isomorphic to colourings found previously using Comer's method. 
We now briefly document our approach to checking isomorphism. 
Sufficient conditions for non-isomorphism enable us to perform a quick and useful preliminary check, as follows. 
Given a clique $K_m=(V,E)$ coloured with colours $c_i\in C$ via the map $c$, we can define its monochromatic subgraph
\[G_i=(V,E_i),\]
with edges
\[E_i=\{\{x,y\}:c(\{x,y\})=c_i\}.\]
The {\sc nauty/Traces} tool can be used to compute the canonical form $\Can(G_i)$ of the uncoloured graph and we can thereby compute the multiset
\[\mathcal{M}(c,K_m) = \bigl[\!\bigl[ \Can(G_1),\ldots,\Can(G_n) \bigr]\!\bigr].\]%%\Bigl\{\bigl\{ \Can(G_1),\ldots,\Can(G_n) \bigr\}\Bigr\}.\]

\newtheorem{proposition}{Proposition}
\begin{proposition}
If two edge-coloured cliques $K_m$ and $K'_m$, coloured with $c$ and $c'$ respectively, are isomorphic allowing relabelling of colours, then $\mathcal{M}(c,K_m)=\mathcal{M}(c',K'_m)$.
\end{proposition}
\begin{proof}
By definition isomorphic graphs have the same canonical label.
A coloured-graph isomorphism consists of a vertex permutation $f$ and a colour permutation $\sigma$. Specifically, we have
\[c'(\{f(x),f(y)\})=\sigma(c(\{x,y\}))\]
for all distinct $x,y\in V(K_m)$.
%% c.g. note I am not using f an g symbols here, as that would overload $g$
%%
For each colour $c_i$, the map $f$ is an ordinary graph isomorphism, so that
\[\Can(G_i)=\Can(G'_{\sigma(i)}).\]

Permutation $\sigma$ only relabels edge colours, therefore multisets of canonical labels are equal. 
\end{proof}

Using this proposition we are able to compute non-isomorphism quickly for the purposes of the reported experiments. 
Naturally, that $\mathcal{M}(c,K_m)=\mathcal{M}(c',K'_m)$ does not in general imply the graphs are isomorphic.
We apply a test using the standard incidence-graph method, using {\sc nauty/Traces} to establish isomorphism where applicable.

An analogous sufficient test uses counts of coloured triangles.
For an $n$-colouring $c$ of $K_m$ edges, let $N_c(i,j,k)$ denote the number of triangles whose unordered edge-colour triple is $\{c_i,c_j,c_k\}$, and define
\[
\mathcal{T}(c,K_m) =
\bigl[\!\bigl[ N_c(i,j,k): 1\leq i\leq j\leq k\leq n \bigr]\!\bigr].
\]

\begin{proposition}\label{prop:trianglecounting}
If two edge-coloured cliques $K_m$ and $K'_m$, coloured with $c$ and $c'$ respectively, are isomorphic allowing relabelling of colours, then
  \[ \mathcal{T}(c,K_m)=\mathcal{T}(c',K'_m). \]
\end{proposition}
\begin{proof}
A coloured-graph isomorphism maps triangles bijectively to triangles.
A permutation of the colour classes merely permutes the colour triples indexing their counts.
Therefore, the multiset of coloured-triangle counts is unchanged.
\end{proof}

At $n=9$ and $p=397$, the admissible index differences are $\{4,5,7,8\}$.
We exhaustively enumerated every perfect matching of the admissible-pair graph $H_{\mathrm{adm}}.$ 
The enumeration recursively selects the least unmatched index and branches over every admissible partner, so every perfect matching is generated exactly once. 
This yields $54{,}352$ perfect matchings.
Here, $18=2n$ is the number of fine cosets. 
Cyclic translation acts on a perfect matching 
\[ M=\bigl\{\{a_1,b_1\},\ldots,\{a_9,b_9\}\bigr\} \] 
by 
\[M+s = \bigl\{ \{a_1+s,b_1+s\},\ldots,\{a_9+s,b_9+s\} \bigr\}, \qquad s\in\Z_{18},\]
%%
%%\[ M+t = \bigl\{ \{a_1+t,b_1+t\},\ldots,\{a_9+t,b_9+t\} \bigr\}, \qquad t\in\Z_{18}, \] 
with all indices reduced modulo $18$. 
Perfect matchings in the same orbit yield isomorphic colourings because multiplication by $g^s$ maps $C_a$ to $C_{a+s}$.
We identify each orbit by retaining the lexicographically least of its eighteen translates.
This yields $3{,}028$ orbits under cyclic translation of $\Z_{18}$.
We computed $\mathcal{T}(c,K_{397})$ for one representative of every such orbit.
The resulting $3{,}028$ multisets were pairwise distinct: no two orbit representatives produced the same multiset of coloured-triangle counts.
%%All $3{,}028$ invariants were distinct.
%%
By Proposition~\ref{prop:trianglecounting}, the corresponding colourings are pairwise non-isomorphic, even allowing permutations of the nine colours.
 
\section{Author Contributions}

Charles Gretton discovered the \FBF\  method.
Richard Pennifold developed {\sc SageMath} scripts implementing Comer's method which are the basis of the first implementation of \FBF.
Richard Pennifold also discovered several new colourings reported in this paper.
Tomasz Kowalski discovered multiple reported colourings, brought the colouring problem to our attention in the first place, and noticed that the \FBF\  method produces colourings that are non-isomorphic to existing colourings (e.g., see the two colourings at $n=5$ and $p=101$ in Table~\ref{tab:results}). 
%%
%% Tomasz also patiently described the problem to Charles Gretton, over many many years, correcting occasionally misunderstandings, sharing GAP-language scripts, and directing us to key literature.
%%
Cody Christopher observed that the monochromatic admissibility of a fused pair depends only on its coset-index difference and proposed the resulting symmetry reduction.
The reduction can permit direct construction in some cases, but in general mixed-difference pairings must still be found by search and verified against the sumset compatibility conditions between distinct fused pairs.
%%
%% Cody also identified such mixed-difference pairings as a source of further non-isomorphic colourings.

%%(POSSIBLE -- Cody Christopher has come to me with some results that I am still in the process of understanding.)

\section{Future Work}

We remain interested in exploring general-purpose combinatorial search algorithms for the optimisation of the size, and for the discovery of combinatorial objects with extensive symmetry. 
Starting from colourings constructed using our method, general-purpose SAT solvers are able to find smaller clique colourings that correspond to Ramsey colourings. 
That topic is beyond the scope of this short paper, which is written predominantly to announce the discovery of colourings for $8$ and $13$ and the corresponding method.  

\bibliographystyle{plain}
\bibliography{papers}

\end{document}